\documentclass{eptcs}
\usepackage[T1]{fontenc}
\usepackage[utf8]{inputenc}
\usepackage[english]{babel}

\usepackage{graphicx}
\usepackage{caption}
\usepackage{subcaption}

\usepackage{amsmath}
\usepackage{amssymb}
\usepackage{mathrsfs}
\usepackage{nicefrac}
\usepackage{breakurl}

\newcommand{\weg}[1]{}

\makeatletter
\providecommand*{\twoheadrightarrowfill@}{%
  \arrowfill@\relbar\relbar\twoheadrightarrow
}
\providecommand*{\xtwoheadrightarrow}[2][]{%
  \ext@arrow 0579\twoheadrightarrowfill@{#1}{#2}%
}
\makeatother

\newcommand{\dom}{\text{dom}}
\newcommand{\CCdots}{\ldots}

\newcommand{\twoheadlongrightarrow}{\relbar\joinrel\twoheadrightarrow}
\newcommand{\rewritesto}{\longrightarrow}
\newcommand{\rewritestoP}{\rewritesto_P}
\newcommand{\eqP}{=_P}
\newcommand{\rewritestoby}[1]{\xrightarrow{#1}}
\newcommand{\rewritestoto}{\twoheadlongrightarrow}
\newcommand{\rewritestotoby}[1]{\xtwoheadrightarrow{#1}}
\newcommand{\normal}{\downarrow}
\newcommand{\twoheaddownarrow}{\mathrel{\rotatebox[origin=c]{270}{$\twoheadrightarrow$}}}
\newcommand{\defined}{\twoheaddownarrow}
\newcommand{\representation}[1]{\left\langle#1\right\rangle}

\newcommand{\nattype}{\mathbf{Nat}} 
\newcommand{\nat}{\nattype} 
\newcommand{\listtype}[1]{\mathbf{List}_{#1}}
\newcommand{\bintreetype}[1]{\mathbf{BinTree}_{#1}}

\newcommand{\Zero}{\mathbf{Zero}}
\newcommand{\Succ}{\mathbf{Succ}}
\newcommand{\AddCBV}{\mathbf{AddCBV}}
\newcommand{\AddCBN}{\mathbf{AddCBN}}
\newcommand{\OMega}{\mathbf{Omega}}
\newcommand{\Id}{\mathbf{Id}}

\newcommand{\StoreNat}{\mathbf{StoreNat}}
\newcommand{\UnstoreNat}{\mathbf{UnstoreNat}}
\newcommand{\UseNat}{\mathbf{UseNat}}
\newcommand{\StoreNatA}{\mathbf{StoreNatA}}
\newcommand{\StoreNatB}{\mathbf{StoreNatB}}

\newcommand{\zeroname}{\Zero}
\newcommand{\succname}{\Succ}

\newcommand{\vc}{{\sf c}}
\newcommand{\vr}{{\sf r}}
\newcommand{\vf}{{\sf f}}
\newcommand{\vx}{{\sf x}}

\newcommand{\hatf}{\hat{f}}
\newcommand{\hatC}{\widehat{(c_1,c_2)}}

\newcommand{\programP}{{\mathbb P}}

\newcommand{\IN}{{\mathbb N}}

\newcommand{\itcbv}[2]{\mathbf{ItCBV}_{#1 \rightarrow #2}}
\newcommand{\itcbn}[2]{\mathbf{ItCBN}_{#1 \rightarrow #2}}

\newcommand{\itcbvimpl}[4]{\itcbv{#1}{#2}^{\mathsf{#3}, #4}}
\newcommand{\itcbnimpl}[3]{\itcbn{#1}{#2}^\mathsf{#3}}

\newcommand{\Names}{\mathscr{N}}
\newcommand{\Vars}{\mathscr{V}}

\newcommand{\oeq}{\approx} 

\newcommand{\constr}[2]{\mathbf{Cons}^{#1}_{#2}}

\newcommand{\paramcount}[2]{m^{#1}_{#2}}
\newcommand{\paramtype}[2]{\phi^{#1}_{#2}}

\newcommand{\conttype}[2]{\psi^{#1}_{#2}}

\usepackage{prooftree}
\newcommand{\notjustify}{\thickness0em\justifies}
\usepackage[figuresleft]{rotating}

\newcommand{\trvar}{\mbox{(\text{Var})}}
\newcommand{\trrule}{(\text{Name})}
\newcommand{\trappl}{(\text{Appl})}
\newcommand{\treqconv}{(=\text{-conv})}
\newcommand{\trnil}{(\text{Nil})}
\newcommand{\trcons}{(\text{Cons})}

\newcommand{\trmuconv}{(\mu\text{-conv})}

\newcommand{\trfconv}{(f\text{-conv})}

\usepackage{cite}

\usepackage{enumerate}

 \newtheorem{definition}{Definition}[section]
 \newtheorem{lemma}[definition]{Lemma}
 \newtheorem{proposition}[definition]{Proposition}
 \newtheorem{theorem}[definition]{Theorem}
 \newtheorem{corollary}[definition]{Corollary}

 \newtheorem{example}[definition]{Example}

 \newtheorem{notation}[definition]{Notation}
 \newtheorem{convention}[definition]{Convention}

 \newenvironment{proof}{\begin{trivlist} \item[\hskip \labelsep%
 {\bf Proof}]}{\end{trivlist}}

\newcommand{\qed}{$\square$}

\weg{
\newcounter{thmcounter}
\newtheorem{lemma}[thmcounter]{Lemma}
\newtheorem{theorem}[thmcounter]{Theorem}
\newtheorem{corollary}[thmcounter]{Corollary}
\newtheorem{example}[thmcounter]{Example}
\newtheorem{definition}[thmcounter]{Definition}

\newtheorem{notation}[thmcounter]{Notation}
\newtheorem{convention}[thmcounter]{Convention}

} 

\usepackage{hyperref}

\usepackage{comment}

\usepackage{float}

\usepackage{parskip}

\newcommand{\con}{\mathbf{C}}

\begin{document}

\title{A type system for Continuation Calculus} \author{Herman Geuvers
  \and Wouter Geraedts \and Bram Geron \and Judith van Stegeren}

\author{Herman Geuvers\institute{Radboud University Nijmegen,
\and
Technical University Eindhoven, the Netherlands}
\and Wouter Geraedts\institute{Radboud University Nijmegen, the Netherlands} 
\and Bram Geron\institute{School of Computer Science, University of Birmingham, UK} 
\and Judith van Stegeren\institute{Radboud University Nijmegen, the Netherlands} 
}
\def\authorrunning{H. Geuvers, W. Geraedts, B. Geron, J. van Stegeren}

\def\titlerunning{A type system for Continuation Calculus}


\maketitle

\begin{abstract}
Continuation Calculus (CC), introduced by Geron and Geuvers \cite{gg},
is a simple foundational model for functional computation. It is
closely related to lambda calculus and term rewriting, but it has no
variable binding and no pattern matching. It is Turing
complete and evaluation is deterministic.  Notions like
``call-by-value'' and ``call-by-name'' computation are available by
choosing appropriate function definitions: e.g.\ there is a
call-by-value and a call-by-name addition function.


In the present paper we extend CC with types, to be able to define
data types in a canonical way, and functions over these data types,
defined by iteration. Data type definitions follow the so-called
``Scott encoding'' of data, as opposed to the more familiar ``Church
encoding''. 

The iteration scheme comes in two flavors: a call-by-value and a
call-by-name iteration scheme. The call-by-value variant is a double
negation variant of call-by-name iteration. The double negation
translation allows to move between call-by-name and call-by-value.


%

\end{abstract}


\section{Introduction}

Continuation calculus (or CC) \cite{gg} is a crossover between term
rewriting systems and $\lambda$-calculus. Rather than focusing on
expressions, continuation calculus treats continuations as its
fundamental object. This is accomplished by restricting evaluation to
strictly top-level, discarding the need for evaluation inside
contexts. This also fixes an evaluation order, so the representation
of a program in CC depends on whether call-by-value or call-by-name is
desired.  Furthermore, CC ``separates code from data'' by placing the
former in a static \emph{program}, which is sourced for reductions on
a term. Variables are absent from terms, and no substitution happens
inside terms.

Despite the obvious differences between CC and $\lambda$-calculus with
continuations (or $\lambda_C$), there seems to be a strong correspondence. For
instance, it has been suggested \cite{geronthesis} that programs in either can
be simulated in the other up to (non)termination, in an untyped setting.
Progress so far suggests that continuation calculus might become a useful
alternative characterization of $\lambda_C$, and theorems in one system could
apply without much effort to the other.

The purpose of this paper is to strengthen the correspondence between
CC and $\lambda$-calculus, by introducing a type system for CC and by
showing how data types and functions over data can be defined in
CC. The type system rejects some undesired terms and the types
emphasize the difference between call-by-name and call-by-value. Also,
the types pave the way for proving properties of the programs. The
types themselves do not enforce termination in general, because the system is
`open': programs are understood to be only parts of a larger whole, and 
names with no rule in a certain program are names whose ``behavior'' is 
left unspecified. However, if the
programs are defined using only iteration and
non-circular program rules, all terms are terminating. This we show
in a separate paper.

\subsection{Informal definition of CC}

Terms in CC are of the shape $n.t_1.t_2.\CCdots.t_k$, where $n$ is a name and $t_i$ is again a term. The `dot' denotes binary application, which is left-associative. In CC, terms can
be evaluated by applying {\em program rules\/} which are of the shape
$$n.x_1.x_2.\CCdots.x_p \rewritesto u, \hfill (*)$$ 
where $u$ is a term over variables $x_1 \ldots x_p$. However, this
rule can only be applied on the `top level': 
\begin{itemize}
\item reduction is not a congruence;
\item rule $(*)$ can only be applied to the term $n.t_1.t_2.\CCdots.t_k$
  in case $k=p$,
\item then this term evaluates to $u[t_1/x_1, \ldots, t_p/x_p]$.
\end{itemize}
CC has no pattern matching or variable
binding, but it is Turing complete and a translation to and
from the untyped $\lambda$-calculus can be defined that preserves and reflects termination, see
\cite{geronthesis}.

In continuation calculus, the natural numbers are represented by the
names $\zeroname$ and $\succname$ and the following two program-rules:
\[
\begin{array}{lcl}
\Zero.c_1.c_2 &\rewritesto & c_1 \\
\Succ.x.c_1.c_2 & \rewritesto & c_2.x \\
\end{array}
\]
So $\Zero$ represents $0$, $\Succ.\Zero$ represents $1$,
$\Succ.(\Succ.\Zero)$ represents $2$ etcetera. This representation of
data follows the so-called Scott encoding, which is known from the
untyped lambda calculus by defining $\Zero := \lambda x\,y. x$, $\Succ
:= \lambda n.\lambda x\, y. y\,n$ (e.g.\ see \cite{abadi,jansen}). The
Scott numerals have ``case-distinction'' built in (distinguishing
between $0$ and $n+1$), which can be used to mimic pattern
matching. The more familiar Church numerals have iteration built
in. For Scott numerals, iteration has to be added, or it can be
obtained from the fixed-point combinator in the case of untyped lambda
calculus. For CC the situation is similar: we have to add iteration
ourselves.

As an example, we define addition in two ways: in call-by-value (CBV) and in
call-by-name (CBN) style (\cite{gg}).

\begin{example}\label{exa.add}
\[
	\begin{array}{lcl}
		\mathbf{AddCBV}.n.m.c & \rewritesto & n.(c.m).(\mathbf{AddCBV'}.m.c) \\
		\mathbf{AddCBV'}.m.c.n' & \rewritesto & \mathbf{AddCBV}.n'.(\Succ.m).c \\
	\end{array}
\]
\[
	\begin{array}{lcl}
		\mathbf{AddCBN}.n.m.c_1.c_2 & \rewritesto & n.(m.c_1.c_2).(\mathbf{AddCBN'}.m.c_2) \\
		\mathbf{AddCBN'}.m.c_2.n' & \rewritesto & c_2.(\mathbf{AddCBN}.n'.m)\\
	\end{array}
\]
For $\mathbf{AddCBV}$ we find that
$\mathbf{AddCBV}.(\Succ^n.\Zero).(\Succ^m.\Zero).c$ evaluates to
$c.(\Succ^{n+m}.\Zero)$: the result of the addition function is
computed completely and passed as argument to the continuation
$c$. For $\mathbf{AddCBN}$, only a first step in the computation is
carried out and then the result is passed to the appropriate
continuation $c_1$ or $c_2$.
\end{example}

Continuation calculus as it occurs in \cite{gg} is untyped.  In the
present work we present a typing system for continuation calculus.
The typing system gives the user some guarantee about the meaning and
well-formedness of well-typed terms.  We also develop a general
procedure for defining algebraic data-types as types in CC and for
transforming functions defined over these data types into valid typed
terms in CC. In a separate paper we prove termination of all
well-formed iterative CC programs \cite{geuvers-typlamIt}.

\section{Formal definition of CC}


For the detailed formal definition, we refer to \cite{gg}. Here we give a short recap of CC.
The \emph{terms} are either a name, or the combination $t.u$ of two terms $t$ and $u$; $\Names$ is any infinite set of \emph{names}. So, the terms do not contain
variables. (One could add them, but it's not necessary.) Names act as
labels for functions and constructors in CC. \emph{Names} in CC start
with an uppercase letter and are printed in \textbf{bold}. The dot is
left-associative, so we write $(((n.t_1).t_2).\CCdots.t_k)$ as
$n.t_1.\CCdots.t_k$.

The \emph{head} of a term is its `leftmost' name:
$\text{head}(n.t_1.t_2.\CCdots.t_k) = n$. The \emph{length} of a term
is the number of dots towards the head:
$\text{length}(n.t_1.t_2.\CCdots.t_k) = k$.

To define programs we assume any infinite set $\Vars$ of \emph{variables}.
A \emph{program} is a set of \emph{rules}, each of the following shape
\[
 n.x_1.x_2.\CCdots.x_k \rewritesto u
\]
where the $x_i$ are distinct variables and $u$ is a {\em term over the
  variables $x_1, \ldots, x_k$}, so $u$ is a term that may use, apart
from names, also the variables $x_1, \ldots, x_k$. We say that the
rule \emph{defines} the name $n$.  Within a program, a name may occur
at most once as the head of a rule. If $P$ is a program, the {\em
  domain of $P$}, $\dom(P)$ is the set of names that is defined in
$P$.

Let a  program $P$ be given.
A term can be \emph{evaluated in $P$\/} by applying one of the
rules of $P$ to the whole term as follows. Suppose $P$ contains the
rule $n.x_1.x_2.\dots.x_k \rewritesto u$, then
	\[ n.t_1.t_2.\dots.t_k \rewritestoP u \left[ t_1/x_1, \CCdots, t_k/x_k \right] \]
where the latter denotes the substitution of $t_1, \CCdots, t_k$ for
$x_1, \CCdots, x_k$. We usually omit the subscript $P$ and just
write $\rewritesto$, as $P$ will be clear from the context.

It should be noted that one does not evaluate `under the
application'. To make this explicit we introduce some more
terminology.  A name $n$ has \emph{arity $k$} in $P$ if $P$ contains a
rule of the form
		\[ n.x_1.x_2.\CCdots.x_k \rewritesto u \]
Similarly, a term $t$ has an \emph{arity} in $P$ if $\text{arity}(\text{head}(t)) \geq \text{length}(t)$ and we define
		\[ \text{arity}(t) := \text{arity}(\text{head}(t)) - \text{length}(t) \]
A term $t$ is \emph{defined in $P$} if $\text{head}(t) \in
\text{dom}(P)$. Otherwise $t$ is \emph{undefined in $P$}.  A defined
term is either \emph{complete}, if $\text{arity}(t) = 0$, or
\emph{incomplete} if $\text{arity}(t) > 0$, or \emph{invalid} if it
has no arity.

We write $A \rewritesto B$ for ``$B$ is a reduct of $A$'' and $A
\rewritestoto B$ for ``$A$ reduces in zero or more steps to $B$''.
Because every name is defined at most once in the set of program
rules, every program is a deterministic rewriting system.

A term $M$ is said to be \emph{terminating} (or strongly normalizing)
in $P$ if there exists a reduct $N$ such that $N$ can no longer be
rewritten using the rules from $P$. Then $N$ can be undefined, incomplete, or invalid. We write
\[
\begin{array}{ll}
M \normal_
P &
	\text{if~} M\text{~can not be rewritten using the rules of~} P \\
M \defined_P &
\text{if~} \exists N(M \rewritestoto N \wedge N \normal_P)
\end{array}
	\]
Note that M $\defined_{P}$ implies $M \rewritestoto N \downarrow_P$, as $M \rewritestoto N$ can mean that $M$ rewrites to $N$ in zero steps.

The simplest notion of equality between terms in CC is the transitive,
symmetric, reflexive closure of $\rewritesto$, which we denote by
$\eqP$. So $M_1 \eqP M_2$ in case there is an $N$ such that $M_1
\rewritestoto_P N$ and $M_2 \rewritestoto_P N$. This is an interesting
equivalence relation, however, it is much too fine, as we show in the
following example. (See also \cite{gg}.)

\begin{example}\label{exa.cbnadd}
For the call-by-name addition of Example\ref{exa.add}, we find that
$$\AddCBN.(\Succ.\Zero).\Zero.c_1.c_2 \rewritestoto
\AddCBN'.\Zero.c_2.\Zero \rewritestoto c_2.(\AddCBN.\Zero.\Zero)$$ If
we also compute $\Succ.\Zero.c_1.c_2$, we obtain $c_2.\Zero$, which is
not the same term, so we don't have
$\AddCBN.(\Succ.\Zero).\Zero.c_1.c_2 \eqP \Succ.\Zero.c_1.c_2$.\\ 
If we also allow computing `under the function $c_2$', the terms are
still not equal: $\AddCBN.\Zero.\Zero$ does not reduce to
$\Zero$. However, when supplied with two continuations, $d_1$ and
$d_2$, they are equal: $\AddCBN.\Zero.\Zero.d_1.d_2 \rewritestoto d_1$
and $\Zero.d_1.d_2 \rewritestoto d_1$.
\end{example}

In the example we see two terms $M$ and $N$ which are `equal for all
practical purposes', but we don't have $M\eqP N$. We say that two terms $M$ and $N$ are
\emph{observationally equivalent} under program $P$, notation $M
\approx_P N$, if for all extension programs $P' \supseteq P$ and all
terms X
\[ X.M \defined_{P'} ~\Longleftrightarrow~ X.N \defined_{P'} \]

We recall some properties about $\approx_P$ from \cite{gg}. Proofs can
be found in \cite{gg}.

\begin{lemma}  
The relation $\approx_P$ is a congruence, that is, if $M_1 \approx_P
M_2$ and $N_1 \approx_P N_2$ , then $M_1.N_1 \approx_P M_2.N_2$.
\end{lemma}

\begin{lemma}\label{lem.approx} Let $M, N$ be terms of arity $k$. If $M.c_1.\ldots.c_k \eqP N.c_1.\ldots.c_k$ for fresh names $c_1, \ldots, c_k$, then $M \approx_P N$.
\end{lemma}

\begin{corollary}\label{cor.approx} 
If $M \eqP N$ and $\text{arity}(M) = \text{arity}(N) = 0$, then $M
\approx_P N$.
\end{corollary}

It is not in general the case that $M \eqP N$ implies $M \approx_P
N$. The reason is that reduction of a term need not respect the arity,
giving rise to undesired situations, as can be seen in the following
example (also taken from \cite{gg}). Our typing system will prevent
these situations to occur.

\begin{example}\label{exa.omega}
Consider the following program rules
\begin{eqnarray*}
\Id.x &\rewritesto & x\\
\OMega.x &\rewritesto & x.x
\end{eqnarray*}
Then $\Id.\OMega\rewritesto \OMega$, which is an incomplete term. If
we append another term to $\Id.\OMega$, it becomes invalid:
$\Id.\OMega.M$ has no arity. On the other hand, $\OMega.\OMega
\rewritesto \OMega.\OMega$, so this term is non-terminating. Hence,
$\Id.\OMega \rewritesto \OMega$, but $\Id.\OMega \not\approx_P
\OMega$.
\end{example}

The type system will prevent situations as in Example \ref{exa.omega},
by making the program rule for $\OMega$ not `well-typed' (and thereby
not allowed).  Also note that $\Id.\OMega \not\approx_P \OMega$ is
only possible because $\text{arity}(\Id.\OMega) \neq
\text{arity}(\OMega)$. The type system will make sure that, if
$M\rewritesto N$, then $N$ also has arity $0$.


\section{Types for Continuation Calculus}


\label{section:typesystem}

\begin{definition}\label{def.CCtypes}
We define \emph{types} of CC as follows.
\[ \text{Type} := \bot ~|~ \text{Var} ~|~ (\text{Type} \rightarrow \text{Type}) ~|~ \mu \text{Var} . \text{Type}.\]
where, in $\mu X.\Phi$, we require $\Phi$ to be of the shape
$$\sigma^1 \rightarrow \ldots \rightarrow \sigma^n \rightarrow \bot \mbox{ (with }n\geq 0\mbox{),}$$
with each $\sigma^i$ of the shape
$$\tau^i_1 \rightarrow \ldots \rightarrow \tau^i_{a_i} \rightarrow
\bot \mbox{ (with }a_i\geq 0\mbox{),}$$ where each $\tau^i_j$ is
either $X $ or does not contain $X $.
\end{definition}

As usual, we leave out the parentheses around functions types, so $A
\rightarrow B \rightarrow C$ always means $A \rightarrow (B
\rightarrow C)$.

As a consequence of the above definition, if we have a type $\mu X
.\Phi(X )$, then $X $ occurs \emph{positively} in $\Phi$. We could
have been more liberal, by allowing {\em all\/} types $\mu X .\Phi(X
)$ where $X $ occurs positively in $\Phi(X )$, but that is unnecessary to interpret first order data-types.

The intention of the recursive type $\mu X  . \Phi(X )$ is that it denotes the
type $\sigma$ for which $\sigma = \Phi(\sigma)$.  
To give the $\mu$-types their semantics, we introduce type equalities.

\begin{definition}\label{def.typeeqn}
We define \emph{equality between types}, $\sigma=\tau$, as the least
equivalence relation that can be derived using the following rules.\\
\begin{tabular}{ccc}
  $\begin{prooftree}
  \:\hspace*{3cm}
  \justifies
  \mu X  . \tau = \tau \left[ \nicefrac{\mu X  . \tau}{X } \right]
  \using\trmuconv
  \end{prooftree}$
  &
\hspace*{3cm}
&
  $\begin{prooftree}
    \[
    \notjustify
    \sigma = \tau
    \]
    \[
    \notjustify
    \upsilon = \rho
    \]
  \justifies
  \sigma \rightarrow \upsilon = \tau \rightarrow \rho
  \using\trfconv
  \end{prooftree}$
\end{tabular}
\end{definition}


\weg{
\begin{definition}\label{def.posneg}
Give a type variable $X $ and a type $\Phi$, the notions {\em
  $X $ occurs positively in $\Phi$\/} and {\em $X $ occurs
  negatively in $\Phi$\/} are mutually defined as follows.
\begin{enumerate}
\item $X $ occurs positively in $\Phi$ if $X  \notin
  \text{FV}(\Phi)$ or $\Phi \equiv X $ or $\Phi \equiv \Phi_1
  \rightarrow \Phi_2$ and $X $ occurs negatively in $\Phi_1$,
  positively in $\Phi_2$.
\item $X $ occurs negatively in $\Phi$ if $X  \notin
  \text{FV}(\Phi)$ or $\Phi \equiv \Phi_1
  \rightarrow \Phi_2$ and $X $ occurs negatively in $\Phi_2$,
  positively in $\Phi_1$.
\end{enumerate}
\end{definition}
} 

For a program rule $n.x_1.\CCdots.x_n \rewritesto u$ with $x_1 :
\tau_1, \CCdots, x_k : \tau_k$, we will define $n$ to have the
following type: $n : \tau_1 \rightarrow \CCdots \rightarrow \tau_k
\rightarrow \bot$.  So $\bot$ will be used as the type of complete
CC-terms. This is very much in line with the approach taken by Miquel
\cite{miquel}.

\weg{
For a program rule
	\[ n.x_1.\CCdots.x_n \rewritesto u \qquad \text{with~} x_1 : \tau_1, \CCdots, x_k : \tau_k ,\]
we will define $n$ to have the following type:
	\[ n : \tau_1 \rightarrow \CCdots \rightarrow \tau_k \rightarrow \bot \]
So the type $\bot$ will be used as the type of complete CC-terms.
}

\weg{
If $X $ occurs positively (resp.\ negatively) in $\Phi$, we also
view $\Phi$ is a \emph{type scheme}, which we indicate by writing
$\Phi(X )$, where $X $ refers to all occurrences of $X $
in $\Phi$. A positive, respectively negative, type scheme $\Phi$ can
be applied to a function $f : \tau \rightarrow \rho$, obtaining
$\Phi(f):\Phi(\rho) \rightarrow \Phi(\tau)$, respectively
$\Phi(f):\Phi(\tau) \rightarrow \Phi(\rho)$ via the usual
\emph{lifting\/} operation.
}
 
CC-types will be printed in \textbf{bold}. For example, the type
representing natural numbers $\mathbb{N}$ will be printed as
$\nattype$. Abstract types (i.e. for any $\sigma \in \text{Type}$) are
printed as $\sigma$, $\tau$, $A$, $B$, etc.

\begin{example}\label{exa.CCdatatypes}
The types in CC of some well-known algebraic data-types
	\[
\begin{array}{lcl}
\mathbf{Bool} &
				:= &
					\bot \rightarrow \bot \rightarrow \bot \\
\nattype &
				:= &
					\mu T.\bot \rightarrow (T \rightarrow \bot) \rightarrow \bot \\
\listtype{A} &
				:= &
					\mu T.\bot \rightarrow (A \rightarrow T \rightarrow \bot) \rightarrow \bot \\
\bintreetype{A} &
				:= &
					\mu T.\bot \rightarrow (A \rightarrow T \rightarrow T \rightarrow \bot) \rightarrow \bot \\
		\end{array}
	\]
\end{example}

\begin{convention}
We make use of the convention in logic to define $\neg A$ as
$A\rightarrow \bot$ to introduce $\neg \sigma$ as an abbreviation for
the type $\sigma \rightarrow \bot$. Similarly, $\neg\neg\sigma$
denotes $(\sigma\rightarrow\bot)\rightarrow \bot$.
\end{convention}


\begin{definition}
A \emph{program signature} $\Sigma$ is a finite set $\mathscr{N}
\times \text{Type}$
	\[ \Sigma = n_1 : \sigma_1, \CCdots, n_p : \sigma_p ~~~(\text{with all~} n_i \text{~distinct})\]
A \emph{typing context} $\Gamma$ is a finite set $\mathscr{V} \times \text{Type}$:
	\[ \Gamma = x_1 : \sigma_1, \CCdots, x_n : \sigma_n ~~~\text{(with all~} x_i \text{~distinct)} \]
\end{definition}

The signature gives the types of the names; it is constructed
specifically for a program $P$. The context is just a ``temporary''
set of variables; contexts will be used to define program rules.

We are interested in two kinds of judgment: 
\begin{enumerate}
\item $\Sigma \vdash P$, a \emph{program judgment}, to express that,
  given a program signature $\Sigma$, $P$ is a well-typed program. So
  $P$ will consist of program rules.
\item $\Gamma \vdash_\Sigma M : A$, a \emph{typing judgment}, to
  express that $M$ -- a term with free variables in $\Gamma$ -- has
  type $A$, given program signature $\Sigma$ and typing context
  $\Gamma$.
\end{enumerate}

\begin{definition}
The derivation rules to derive \emph{typing judgments} are the following
\begin{center}
	\begin{tabular}{rlcrl}
\trvar &
$\begin{prooftree}
x : \sigma \in \Gamma
\justifies
\Gamma \vdash_\Sigma x : \sigma 		
\end{prooftree}
			$ &&
\trrule&
$\begin{prooftree}
n : \sigma \in \Sigma
\justifies
\Gamma \vdash_\Sigma n : \sigma 
\end{prooftree}
			$ \\
&&&&			\\
\trappl&			
$\begin{prooftree}
				\[
					\notjustify
					\Gamma \vdash_\Sigma M : \sigma \rightarrow \tau
				\]
				\[
					\notjustify
					\Gamma \vdash_\Sigma N : \sigma
				\]
				\justifies
					\Gamma \vdash_\Sigma M . N : \tau
			\end{prooftree}
			$ &&
\treqconv&			
$\begin{prooftree}
				\[
					\notjustify
					\Gamma \vdash_\Sigma M : \sigma
				\]
				\[
					\notjustify
					\sigma = \tau
				\]
				\justifies
					\Gamma \vdash_\Sigma M : \tau
			\end{prooftree}
			$
			\\
		\end{tabular}
	\end{center}
\end{definition}

\begin{definition}\label{def.well-typed}
The derivation rules to derive \emph{program judgments} are the following	
\begin{eqnarray*}
\trnil &&
	\begin{prooftree}
		\justifies
			\Sigma \vdash \emptyset
	\end{prooftree}\\
&&\\
\trcons	&&\begin{prooftree}
		\Sigma \vdash P
	\;\;\;\;
	x_1 :A_1, \ldots, x_k : A_k \vdash_\Sigma q : \bot
	\;\;\;\;
n : A_1 \rightarrow \ldots \rightarrow A_k \rightarrow \bot ~\in~ \Sigma
		\justifies
			\Sigma \vdash P \cup \left\{ n . x_1 . \CCdots . x_k \rewritesto q \right\}
		\using{\text{if~} n \text{~not defined in~} P}
	\end{prooftree}
\end{eqnarray*}
We say that {\em program $P$ is well-typed in $\Sigma$\/} in case
$\Sigma \vdash P$. Usually, $\Sigma$ will be clear and we just say
that {\em $P$ is well-typed}. Similarly, we say that the program rule
$n . x_1 . \CCdots . x_k \rewritesto q$ is {\em well-typed in $P$\/} in
case $P \cup \left\{ n . x_1 . \CCdots . x_k \rewritesto q \right\}$
is well-typed.
\end{definition}

The second and third premise in the \trcons\ rule say that the types
of $n . x_1 . \CCdots . x_k$ and $q$ should be both $\bot$. This
guarantees that we can only rewrite terms of type $\bot$.

\begin{example}
\begin{enumerate}
\item Recall the term $\Zero$ with rule $\Zero.c_1.c_2 \rewritesto
  c_1$. We easily verify that this rule is well-typed if we let
  $\Zero : \bot \rightarrow (\nattype \rightarrow\bot)\rightarrow
  \bot$, i.e.\ $\Zero : \nattype$.
\item Similarly, recall the rule for $\Succ$: $\Succ.x.c_1.c_2
  \rewritesto c_2.x$. It is well-typed if we let $\Succ : \nattype
  \rightarrow \bot \rightarrow (\nattype \rightarrow\bot)\rightarrow
  \bot$, i.e.\ $\Succ : \nattype\rightarrow\nattype$.
\item Recalling the definition of $\AddCBV$ in Example \ref{exa.add},
  we see that $\AddCBV$ takes arguments of type $\nattype$, $\nattype$
  and $\nattype \rightarrow \bot$ to produce a term of type $\bot$. So
  the rule is well-typed if we take $\AddCBV : \nattype \rightarrow
  \nattype \rightarrow \neg\neg\nattype$. (NB. $\AddCBV' : \nattype
  \rightarrow (\neg\nattype) \rightarrow \nattype \rightarrow \bot$.)
 \item For the definition of $\AddCBN$ in Example \ref{exa.cbnadd}, we
   see that $\AddCBN: \nattype \rightarrow \nattype \rightarrow
   \nattype$.
\item To type the rule for $\OMega$ in Example \ref{exa.omega}, we
  need $\OMega : \sigma \rightarrow \bot$ with $\sigma =
  \sigma\rightarrow \bot$. But there is no type $\sigma$ for which
  $\sigma = \sigma\rightarrow \bot$, so the rule for $\OMega$ is not
  well-typed.
\end{enumerate}
\end{example}

We have the following properties.

\begin{lemma}
\begin{enumerate}
\item {[}Substitution]
If $n:\tau \in \Sigma$, $\vdash_{\Sigma} t : \sigma$ and
$\vdash_{\Sigma} q:\tau$, then $\vdash_{\Sigma} t [q/n] :\sigma$.
\item {[}Subject reduction] If $\vdash_{\Sigma} t : \sigma$, and
  $t\rewritesto p$, then $\sigma = \bot$ and $\vdash_{\Sigma} p :\bot$.
\end{enumerate}
\end{lemma}

\begin{proof}
\begin{enumerate}
\item By induction on the derivation of $\vdash_{\Sigma} t : \sigma$.
\item Using the first property. If $t\rewritesto p$ by the rule $n
  . x_1 . \CCdots . x_k \rewritesto q$, then $t=n . t_1 . \CCdots
  . t_k$ and $p = q[t_1/x_1, \ldots, t_k / x_k]$.

We have $x_1 :A_1, \ldots, x_k : A_k \vdash_\Sigma q : \bot$ and $x_1
:A_1, \ldots, x_k : A_k \vdash_\Sigma n . x_1 . \CCdots . x_k : \bot$,
so by substitution (where we now substitute terms for variables, but
the argument is the same) we have $t:\bot$ and $p:\bot$.  \qed
\end{enumerate}
\end{proof}

We have the following Corollary of the above and of Corollary \ref{cor.approx}.
 
\begin{corollary}\label{cor.eqPapprox}
If $M$ and $N$ are well-typed terms of type $\bot$ and $M \eqP N$,
then $M \approx_P N$.
\end{corollary}

\subsection{Data types in CC}
\label{sec:constructors}

We have seen the definitions of the types of
booleans, natural numbers and lists in Example
\ref{exa.CCdatatypes}. Here we give a general way of defining
constructors and first order algebraic data types in CC. (That is, for
now we don't allow higher order types in the constructor types.)

\begin{definition}\label{def.datatypes}
A {\em first order data type\/} will be written as
$$\begin{array}{rclrcl}
{\bf data-type} &D& \multicolumn{4}{l}{\mbox{with constructors}}\\
&&& \con^D_1 &:& D^1_1 \rightarrow \ldots \rightarrow  D^1_{a_1} \rightarrow D\\
&&&     & & \ldots\\
&&& \con^D_n &:& D^n_1 \rightarrow \ldots \rightarrow D^n_{a_n} \rightarrow D
\end{array}$$ 
where each of the $D^i_j$ is either $D$ or a type expression that does
not contain $D$. If $D$ is clear from the context, we will omit it as
a superscript and write $\con_i$ instead of $\con^D_i$.
\end{definition}

This defines an algebraic data-type $D$ with $n$ constructors with
names $\con_1, \ldots, \con_n$. Each constructor $\con_i$
has arity $a_i$, which can also be $0$, and then the constructor is a
constant.  

\begin{convention} 
To simplify notation later, we abbreviate the list of argument
types of a constructor, writing $D^1$ for $D^1_1 \ldots D^1_{a_1}$
etc, in a style similar to uncurrying.
\end{convention}

For every constructor we will introduce a name in CC and a rule that
defines it. This program rule acts as a destructor of $D$. If a term
of type $D$ has constructor $\constr{D}{i}$ as its head, all the
arguments of that constructor $t^i_1, \ldots, t^i_{a_i}$ will be
returned to the corresponding continuation $c_i$.

\begin{example}\label{exa.listCC}
Consider the algebraic data type of lists over a type $A$, $\listtype{A}$:
%
In Example \ref{exa.CCdatatypes}, we have defined this data-type in CC
$\mu T.\bot \rightarrow (A \rightarrow T \rightarrow \bot) \rightarrow
\bot $.  The constructors for lists are added to CC by introducing the
following program rules to our program, where $\textbf{Nil} : \bot$
and $\textbf{Cons} : A \rightarrow \listtype{A} \rightarrow \bot$.
\[
	\begin{array}{rcl}
		\textbf{Nil} . c_1 . c_2 &
		\longrightarrow &
		c_1 \\
		
		\textbf{Cons} . x_1 . x_2 . c_1 . c_2 &
		\longrightarrow &
		c_2 . x_1 . x_2 \\
	\end{array}
\]
\end{example}

We now give the general definition of first order data-type in CC.

\begin{definition}\label{def.CCdatatypes}
Given a first order data type $D$ as in Definition \ref{def.datatypes} with $n$ constructors, where, for
$1\leq i \leq n$, $\con_{i} : D^i_1 \rightarrow \ldots D^i_{a_i} \rightarrow D,$ we define the following type $D$ in CC.
$$D := \mu X . (D^1[X /D] \rightarrow \bot) \rightarrow \ldots \rightarrow (D^n[X /D]\rightarrow \bot) \rightarrow \bot.$$
For $i \in [1\ldots n]$,  we add the following constructor $\constr{D}{i}$ to the signature $\Sigma$.
$$\constr{D}{i} : D^i \rightarrow D.$$ Finally, we add for each $i$
($1\leq i\leq n$) the following program rule that acts as a destructor
for $D$.
\[
		\constr{D}{i} . x^i_1 . \ldots . x^i_{a_i} . c_1
                . \ldots . c_n \longrightarrow c_i . x^i_1 . \ldots
                . x^i_{a_i}
	\]
\end{definition}

So, in CC we always have $\text{arity}(\constr{D}{i}) = a_i +
n$. Example \ref{exa.listCC} conforms with this definition. The
constructors are well-typed in CC because we have the equation
$$D = (D^1 \rightarrow \bot) \rightarrow \ldots \rightarrow (D^n\rightarrow \bot) \rightarrow \bot.$$

\begin{notation}\label{not.representation}
Let $D$ be a data type and $\mathbf{D}$ its representation as a type
in Continuation Calculus. If $d:D$ (so $d$ is a data type element of
$D$), we denote by $\representation{d} : \mathbf{D}$ the encoding of
$d$ as a term in CC. (So $\representation{d}$ is defined in the
canonical way using the constructors of Definition
\ref{def.CCdatatypes}.)
\end{notation}

\weg{

To illustrate how the constructors and destructors work in practice
and to give an example of the typing system
we will now take a look at the types of
destructors.

\begin{example}
For the algebraic data type $\listtype{A}$, we have:
\begin{eqnarray*}
\textbf{Nil}		&
			: &
				\bot \rightarrow (A \rightarrow \listtype{A} \rightarrow \bot) \rightarrow \bot \\
\textbf{Nil} . c_1 . c_2 &
			\longrightarrow &
				c_1 \\
\textbf{Cons}		&
			: &
				A \rightarrow \listtype{A} \rightarrow \bot \rightarrow (A \rightarrow \listtype{A} \rightarrow \bot) \rightarrow \bot \\
\textbf{Cons} . x_1 . x_2 . c_1 . c_2 &
			\longrightarrow &
				c_2 . x_1 . x_2
\end{eqnarray*}
\end{example}

A rule that defines constructor name $\constr{\sigma}{i}$ of a type $\sigma$ with $n$ constructors
\[
	\constr{\sigma}{i} . x_{i, 1} . \CCdots . x_{i, \paramcount{\sigma}{i} } . c_1 . \CCdots . c_n \longrightarrow c_i . x_{i, 1} . \CCdots . x_{i, \paramcount{\sigma}{i} }
\]

is of type
\[
	\underbrace{
		\paramtype{\sigma}{i, 1} \rightarrow \CCdots \rightarrow \paramtype{\sigma}{i, \paramcount{\sigma}{i}}
	}_\text{(1)}
	\rightarrow 
		\underbrace{
			(\paramtype{\sigma}{1, 1} \rightarrow \CCdots \rightarrow \paramtype{\sigma}{1, \paramcount{\sigma}{1}} \rightarrow \bot)
		}_{(2_1)}
		\rightarrow \CCdots \rightarrow
		\underbrace{
			(\paramtype{\sigma}{n, 1} \rightarrow \CCdots \rightarrow \paramtype{\sigma}{n, \paramcount{\sigma}{n}} \rightarrow \bot)
		}_{(2_n)}
	\rightarrow \underbrace{\bot}_\text{(3)}
\]
}

\begin{convention}
Unless otherwise specified, $D^i_j$ is the type of the $j^\text{th}$
argument of the $i^\text{th}$ constructor of data type $D$. In the
case of $\listtype{A}$: $D^2_1 = A$, $D^2_2 = \listtype A$.
\end{convention}

\weg{
We can divide the type of a rule for $\constr{D}{i}$ in three parts:
\begin{enumerate}
\item[(1)] The types of the arguments of constructor $\constr{D}{i}$:
  $D^i_1 \ldots D^i_{a_i}$.
\item[(2)] The types of continuations $c_1, \ldots , c_n$, where
  each continuation is a function from the arguments of its
  corresponding constructor to $\bot$, so $c_j : D^j \rightarrow \bot$.
\item[(3)] We have defined $\bot$ as the type of all complete terms. A
  complete term $t$ is a term with exactly the right number of
  parameters, such that $\text{arity}(t) = 0$. So a program rule
  should eventually yield $\bot$.
\end{enumerate}
}
\weg{
\begin{convention}
	\label{convention:conttype}
	Unless otherwise specified $\conttype{\sigma}{i}$ is the type of the $i^\text{th}$ continuation of type $\sigma$:
	\[
		\conttype{\sigma}{i} = \paramtype{\sigma}{i, 1} \rightarrow \CCdots \rightarrow \paramtype{\sigma}{i, \paramcount{\sigma}{i}} \rightarrow \bot
	\]
\end{convention}
}

We often give the typing of terms via a derivation rule.
\[
	\begin{prooftree}
		\[
			\notjustify
			x_{1} : D^i_1
		\] \CCdots
		\[
			\notjustify
			x_{a_i} : D^i_{a_i}
		\]
		\[
			\notjustify
			c_1 : D^1 \rightarrow \bot
		\] \CCdots
		\[
			\notjustify
			c_n : D^n \rightarrow \bot
		\]
		\justifies
		\constr{D}{i}.x_{1}.\CCdots.x_{a_i}.c_1.\CCdots.c_n : \bot
	\end{prooftree}
\]


\subsection{Iteration schemes}

In this section we give iteration schemes for continuation calculus
that provides general mechanisms for defining functions by
recursion. An iteration scheme defines recursive functions in a
general way, ensuring well-definedness and termination for these
functions. In CC we have a call-by-name and a call-by-value
variant of the iteration scheme.

\weg{
The well-known Church encodings of data types in (untyped or typed)
lambda calculus basically {\em are\/} the iteration schemes
themselves. As we have our data in CC based on a Scott encoding, we
have to add an iteration scheme to get (a limited form of) recursion.
}
The simplest and most well-know form of iteration is over $\mathbb{N}$: Given $b : B$, $f : B \rightarrow B$, the function $\mathbf{It}(b,f) :
\mathbb{N} \rightarrow B$ {\em defined by iteration from $b$ and $f$},
is given by
\[
		\mathbf{It}(b,f)(n) = \left\{
			\begin{array}{ll}
				b & \text{if~} n = 0 \\
				f(\mathbf{It}(b,f)(m)) & \text{if~} n = m+1
			\end{array}
		\right. 
	\]

\weg{
\begin{definition}\label{def:itnat}
Given $b : B$, $f : B \rightarrow B$, the function $\mathbf{It}(b,f) :
\mathbb{N} \rightarrow B$ {\em defined by iteration from $b$ and $f$},
is given by
\[
		\mathbf{It}(b,f)(n) = \left\{
			\begin{array}{ll}
				b & \text{if~} n = 0 \\
				f(\mathbf{It}(b,f)(m)) & \text{if~} n = m+1
			\end{array}
		\right. 
	\]
Given $b : B$ and $f : A \rightarrow B \rightarrow B$, the function
$\mathbf{It}(b,f) : \listtype{A} \rightarrow B$ {\em defined by
  iteration from $b$ and $f$}, is given by
\[
		\mathbf{It}(b,f)(l) = \left\{
			\begin{array}{ll}
				b & \text{if~} l = \mathbf{Nil} \\
				f(a, \mathbf{It}(b,f)(k)) & \text{if~} l = \mathbf{Cons}~a~k
			\end{array}
		\right. 
	\]
\end{definition}

A whole variety of functions fits in the iteration scheme: addition,
multiplication, exponentiation and even the Ackermann function if one
allows $B$ to be a function type.

Assume we want to make an iterator for a data-type $D$ with $n$
constructors.  The iterator then takes functions $f_1 \ldots f_n$ and
some value $v \in D$.  The iterator will destruct $v$, yielding the
constructor $c_i$ applied to arguments $\vec{t}$. Consequently it will
apply the corresponding function $f_i$, using the relevant arguments.
If an argument is again of type $D$, the iterator will apply itself
recursively to that argument.
}

\weg{
Iteration schemes introduce mechanics for computation in CC.
For all $\tau, \sigma \in \mbox{Type}$ and $\tau = \mu Y.\phi[Y]$ with constructors $\constr{\tau}{1}, \ldots, \constr{\tau}{n}$ and corresponding arguments $\vec{x}_{1,[1,\paramcount{\tau}{1}]}, \ldots, \vec{x}_{n,[1,\paramcount{\tau}{n}]}$ we can construct an iteration scheme (see below).
}

An iterator $\mathbf{It}_{D}$ for a general data-type $D$ (following the
general scheme for first order data-types in Definition
\ref{def.datatypes}) to some type $B$ has the following type:
$$\prooftree
f_1 : D^1[B/D] \rightarrow B \;\;\;\ldots\;\;\; f_n : D^n[B/D]\rightarrow B
\justifies
\mathbf{It}_D\,f_1\ldots f_n : D \rightarrow B
\endprooftree$$
with
$$\mathbf{It}_D\,f_1\ldots f_n (\con_i\, v_1 \ldots v_{a_i}) = f_i\, V_1\ldots V_{a_i},$$
where $V_j = \mathbf{It}_D\,f_1\ldots f_n\, v_j$ if $v_j:D$ and $V_j =
v_j$ otherwise.

This is not yet the correct type for an iteration scheme in CC. We do
not yet have any continuations as parameters.  We will provide
separate CBN and CBV iteration schemes below.

\weg{
\begin{notation}
	Instead of $x_1.x_2.x_3.x_4$ we will write $\vec{x}_{[1,4]}$.
	Schematically we use the following abbreviations:
	\[
		\begin{array}{lcl}
			\vec{(-)}_{[i,j]} & = & (-)_i . \CCdots . (-)_j \\
			\vec{(-)}_{k,[i,j]} & = & (-)_{k,i} . \CCdots . (-)_{k,j}
		\end{array}
	\]
\end{notation}
}

\subsection{Call-by-name iterators}
For a call-by-name iterator for a data-type $D$, we also have to
consider the return data-type $B$. CBV calculates the entire return value,
but for CBN it is enough to return the proper continuations with the
proper parameters after calculating only one step in the recursion.
So the CBN-iterator also passes around the continuations to the resulting
values. If result type $B$ has $m$ constructors, then the iterator $\itcbn{D}{B}$ also
needs $m$ continuations as arguments.  This differs from a call-by-value iterator,
where we only have one continuation. 

Let in the following, $D$ be a data-type with $n$ constructors ($\con^D_1,\ldots,\con^D_n$) and $B$
a data-type with $m$ constructors.

\begin{definition}\label{def.CCitcbn}
We define the \emph{call-by-name iterator for type $D$ to type $B$} as
follows. We first give the types of the new names. We abbreviate $f_1
\ldots f_n$ to $\vec{f}$, $c_1 \ldots c_m$ to $\vec{c}$ and $x_1
\ldots x_{a_i}$ to $\vec{x}$.
$$\begin{prooftree}
f_1 : D^1[B/D]\rightarrow B \ldots   f_n : D^n[B/D]\rightarrow B 
\;\;\;\;
 c_1 : B^1\to\bot \ldots   c_m : B^m\to\bot
\;\;\;\;   x :D
\justifies
\itcbn{D}{B} . \vec{f} . x . \vec{c} :\bot
\end{prooftree}$$

$$\begin{prooftree}
f_1 : D^1[B/D]\rightarrow B \ldots   f_n : D^n[B/D]\rightarrow B 
\;\;\;\;
 c_1 : B^1\to\bot \ldots   c_m : B^m\to\bot
\;\;\;\;   x_1 :D^i_1 \ldots x^i_{a_i} :D^i_{a_i}
\justifies
\itcbnimpl{D}{B}{i} . \vec{f} . \vec{c} . x_1 \ldots x_{a_i} :\bot
\end{prooftree}$$
The program  rules are
\[
\itcbn{D}{B} . \vec{f} . x . \vec{c}
	\rewritesto
x . (\itcbnimpl{D}{B}{1} . \vec{f} . \vec{c}) . \CCdots . (\itcbnimpl{D}{B}{n} . \vec{f} . \vec{c}),
	\]
and for $i \in[1\ldots n]$:
\[
\itcbnimpl{D}{B}{i} . \vec{f} . \vec{c} . \vec{x}
			\rewritesto
f_i . b(x_{1}) . \CCdots . b(x_{a_i}) . \vec{c}
	\]
	\[
\text{with~} b(x) = \left\{
\begin{array}{ll}
\itcbn{D}{B} .\vec{f} . x & \mbox{if~} x : D \\
			x & \mbox{otherwise} \\
			\end{array}
		\right.
	\]
\weg{
Note that there are $n$ program rules with name $\mathbf{ItCBN}^i_{D \rightarrow B}$ for $1 \le i \le n$, where $n$ is the number of constructors of type $D$.
}
\end{definition}

\weg{
In order to explain what the type of this iterator scheme is, we dissect the parameters:
\begin{enumerate}
	\item $f_1, \CCdots, f_n$ are the same functions as defined in the iteration scheme. The type of such a function $f_i$ is:
		\[
			\begin{array}{lrcl}
				f_i : \theta^{\tau \rightarrow B}_{i} &
					\text{with~} \theta^{\tau \rightarrow B}_{i} &
						 = &
						 	X _{i,1} \rightarrow \CCdots \rightarrow X _{i,\paramcount{\tau}{i}} \rightarrow \sigma \\
				\\
				&
					\text{and~} X _{i,j} &
						= &
							\left\{
								\begin{array}{ll}
									\sigma &
									\mbox{if~} \phi_{i,j} = \tau\\
		
									\phi_{i,j} &
									\mbox{otherwise} \\
								\end{array}
							\right.
			\end{array}
		\]
	\item $x$ is the input, with $x : \tau$.
	\item $c_1, \CCdots, c_p$ are the continuations for the different constructors of $\sigma$. The type of such a continuation $c_i$ is:
		\[
			c_i : \conttype{\sigma}{i} \qquad \text{with~} \conttype{\sigma}{i} = \paramtype{\sigma}{i, 1} \rightarrow \CCdots \rightarrow \paramtype{\sigma}{i, \paramcount{\sigma}{i}} \rightarrow \bot \ \text{(as defined in Convention \ref{convention:conttype})}
		\]
\end{enumerate}

Combined, this translates to the following type:
\[
	\itcbn{\tau}{\sigma} : \theta^{\tau \rightarrow \sigma}_{[1,n]} \rightarrow \tau \rightarrow \psi^\sigma_{[1,p]} \rightarrow \bot
\]
}

In Section \ref{examples}, we give in Example \ref{ex:listcbn} the
call-by-name iterator for $\listtype{A}$ to $\nattype$. The following
can easily be checked. (See Definition \ref{def.well-typed} for the
formal definition of well-typed rules.)

\begin{lemma}
The rules given in Definition \ref{def.CCitcbn} are well-typed.
\end{lemma}

\subsection{Call-by-value iterators}

Call-by-value iterators differ from their call-by-name cousins in the
sense that the result of the computation is `normalized' or fully
evaluated at the end of the computation.

\begin{definition}
\label{def.CCitcbv}
We define the \emph{call-by-value iterator for a type $D$ to $B$} as
follows. (We abbreviate $f_1 \ldots f_n$ to $\vec{f}$.)
$$\begin{prooftree}
  f_1 : D^1[B/D]\rightarrow \neg\neg B \ldots   f_n : D^n[B/D]\rightarrow \neg\neg B 
\;\;\;\;
  c : \neg B
\;\;\;\;   d :D
\justifies
  \itcbv{D}{B}\, \vec{f} \, c \, d : \bot
\end{prooftree}$$
and for $i \in [1, n]$ and $j \in [1, a_i]$, under the same typing hypotheses for $\vec{f}$ and $c$:
$$\begin{prooftree}
x_j:D^i_{j}\ldots x_{a_i}:D^i_{a_i}
\;\;\;\; r_1 :D^i_{1}[B/D] \ldots r_{j-1} :D^i_{j_1}[B/D]
\justifies
\itcbvimpl{D}{B}{i}{j} .\vec{f}. c. x_j\ldots x_{a_i}.r_1\ldots r_{j-1} :\bot
\end{prooftree}$$

The program rules are
\[
\begin{array}{lcl}
\itcbv{D}{B} . \vec{f} . c . x &
			\rewritesto &
x . (\itcbvimpl{D}{B}{1}{1} . \vec{f} . c) . \CCdots . (\itcbvimpl{D}{B}{n}{1} . \vec{f} . c) \\
\end{array}
\]
and for $i \in [1, n]$ and $j \in [1, a_i]$:
\[
\begin{array}{lcl}
\itcbvimpl{D}{B}{i}{j} .\vec{f}. c. x_j\ldots x_{a_i}.r_1\ldots r_{j-1} &
\rewritesto &
\text{LHS} \vspace{1.5ex} \\
  \multicolumn{3}{c}{\text{LHS} =
  \left\{\begin{array}{ll}
\itcbv{D}{B} . \vec{f}. ( \itcbvimpl{D}{B}{i}{j+1} . \vec{f}. c . x_{j+1}  \ldots  x_{a_i} . r_1 \ldots r_{j-1} ) . x_{j} &
		\text{if~} x_{j} : D \\
 \itcbvimpl{D}{B}{i}{j+1} . \vec{f} . c . x_{j+1}\ldots x_{a_i} . r_{1}\ldots r_{j-1} . x_{j} &
                \mbox{otherwise} \\
 \end{array}\right.
	} \vspace{1.5ex} \\
\itcbvimpl{D}{B}{i}{a_i+1} . \vec{f} . c . r_1 \ldots r_{a_i} &
				\rewritesto &
					f_i . r_1 \ldots r_{a_i} . c \\
		\end{array}
	\]
\end{definition}

\weg{
To explain what the type of this iterator scheme is, we dissect the parameters:
\begin{enumerate}
	\item $f_1, \CCdots, f_n$ are the same functions as defined in the iteration scheme, except that they directly apply the CBV-continuation. The type of such a function $f_i$ is:
		\[
			\begin{array}{lrcl}
				f_i : \theta^{D \rightarrow B}_{i} &
					\text{with~} \theta^{D \rightarrow B}_{i} &
						 = &
						 	X _{i,1} \rightarrow \CCdots \rightarrow X _{i,\paramcount{D}{i}} \rightarrow (B \rightarrow \bot) \rightarrow \bot \\
				\\
				&
					\text{and~} X_{i,j} &
						= &
							\left\{
								\begin{array}{ll}
									B &
									\mbox{if~} \phi_{i,j} = D\\
		
									\phi_{i,j} &
									\mbox{otherwise} \\
								\end{array}
							\right.
			\end{array}
		\]
	\item $c$ is the sole continuation, which is of type $B \rightarrow \bot$.
	\item $x$ is the input, with $x : D$.
\end{enumerate}

Combined, this translates to the following type:
\[
	\forall B \in \text{Type~} . \itcbv{D}{B} : \theta^{D \rightarrow B}_{[1,n]} \rightarrow (B \rightarrow \bot) \rightarrow D \rightarrow \bot
\]
}

The technical subtlety in the call-by-value reduction rule lies in
the fact that, in case data-type $D$ has a constructor with more than
one recursive sub-term (e.g.\ in the case of binary trees, where we
have `join', taking two sub-trees), we have to evaluate {\em all
  recursive sub-terms}. The reduction rule makes sure that we do that
and reduce to a complete value before calling the function. 
The following lemma helps in better understanding the terms
$\itcbvimpl{D}{B}{i}{j} .\vec{f}. c.$ in Definition \ref{def.CCitcbv}.

\begin{lemma}\label{lem.corrcbvit}
For $j\in [1\ldots a_i]$, given $x_j:D^i_{j}\ldots x_{a_i}:D^i_{a_i}$
and $r_1 :D^i_{1}[B/D] \ldots r_{j-1} :D^i_{j_1}[B/D]$, the reduct of
$\itcbvimpl{D}{B}{i}{j} .\vec{f}. c. x_j\ldots x_{a_i}.r_1\ldots
r_{j-1}$ is of type $\bot$.
\end{lemma}

\begin{proof}
For $j=a_i$, the result is immediate, for other $j$, the result
follows from the result for $j+1$, making a case distinction between
$D^i_j[B/D] = D^i_j$ or $D^i_j[B/D] = B$. \qed
\end{proof}

The following now easily  follows.

\begin{lemma}
The rules given in Definition \ref{def.CCitcbv} are well-typed.
\end{lemma}

\begin{example}	For the iterator from $\nattype$ to $\nattype$, this amounts to the following
\begin{align*}
\itcbv{\nattype}{\nattype} . f_1 . f_2 . c . x & \rewritesto
x . (\itcbvimpl{\nattype}{\nattype}{Zero}{1} . f_1 . f_2 . c) . (\itcbvimpl{\nattype}{\nattype}{Succ}{1} . f_1 . f_2 . c) \\
\itcbvimpl{\nattype}{\nattype}{Zero}{1} . f_1 . f_2 . c & \rewritesto f_1 . c \\
\itcbvimpl{\nattype}{\nattype}{Succ}{1} . f_1 . f_2 . c . x_1 & \rewritesto
\itcbv{\nattype}{\nattype} . f_1 . f_2 . (\itcbvimpl{\nattype}{\nattype}{Succ}{2} . f_1 . f_2 . c) . x_1 \\
\itcbvimpl{\nattype}{\nattype}{Succ}{2} . f_1 . f_2 . c . r_1 & \rewritesto f_2 . r_1 . c
\end{align*}

This can be compressed a bit if we replace
$\itcbvimpl{\nattype}{\nattype}{Zero}{1} . f_1 . f_2 . c$ by $f_1
. c$.

Another simplification that we can do is to replace some auxiliary
names that are introduced in the iteration scheme by a
$\lambda$-term. For example we can replace
$\itcbvimpl{\nattype}{\nattype}{Succ}{2}$ by the new `name' $(\lambda
\vf_1, \vf_2, \vc , \vr_1 \mapsto \vf_2 . \vr_1 . \vc)$. The convention for such a name is that
$$(\lambda \vf_1, \vf_2, \vc , \vr_1 \mapsto \vf_2 . \vr_1 . \vc).f_1.f_2.c.r_1 \rewritesto f_2.r_1.c.$$ 
So, the arity of the new name is the number of arguments of the $\lambda$ and its program rule is given by the body.
Now the rules for $\itcbv{\nattype}{\nattype}$ simplify to
\begin{align*}
\itcbv{\nattype}{\nattype} . f_1 . f_2 . c . x & \rewritesto
x . (f_1  . c) . (\itcbvimpl{\nattype}{\nattype}{Succ}{1} . f_1 . f_2 . c) \\
\itcbvimpl{\nattype}{\nattype}{Succ}{1} . f_1 . f_2 . c . x_1 & \rewritesto
\itcbv{\nattype}{\nattype} . f_1 . f_2 . ((\lambda
\vf_1, \vf_2, \vc , \vr_1 \mapsto \vf_2 . \vr_1 . \vc).f_1.f_2.c) . x_1
\end{align*}
\end{example}

In Section \ref{examples}, we show more examples, notably in
Example \ref{ex:listcbv} we give the call-by-value iterator for
$\listtype{A}$ and we show how to program the `length' function with
it.

\subsection{Rules for programming with data types in CC}
Starting from the constructors for first order data types and the
call-by-name and call-by-value iterators we can program new functions
from existing ones. However, due to the fact that we are in CC and not
in $\lambda$-calculus, we need some additional `glue' to make flexible
use of the iteration scheme to define functions.

\begin{example}\label{exa.itcbv}
Given $\itcbv{\nattype}{\nattype}$ we can define $\AddCBV$ as follows.
$$\mathbf{AddCBV}.m.n.c := \itcbv{\nattype}{\nattype}. (F_1.m).F_2.c.n$$
where $F_1$ and $F_2$ are defined by
\begin{eqnarray*}
F_1.x.c &\rewritesto&  c.x\\ 
F_2.x.c &\rewritesto&  c.(\Succ.x)
\end{eqnarray*}
So, we need $2$ auxiliary functions to define $\AddCBV$ in terms of
$\itcbv{\nattype}{\nattype}$. In terms of Example \ref{exa.itcbv}, we
need the names $(\lambda \vx,\vc \mapsto \vc.\vx)$, which is $F_1$ and
$(\lambda \vx,\vc\mapsto \vc.(\Succ.\vx))$, which is $F_2$.
\end{example}

The example shows that, to really profit from the expressivity of the
iteration schemes, we must allow the addition of `simple'
functions. These are functions that have a non-circular definition.

\begin{definition} 
A {\em non-circular\/} program rule is a rule of the form
$$n.x_1.\ldots x_k \rewritesto q,$$ where the names occurring in $q$
are restricted to the constructors (Definition \ref{def.CCdatatypes})
and the iterators (CBN, Definition \ref{def.CCitcbn} and CBV,
Definition \ref{def.CCitcbv}).

We define $\programP$ as the set of program rules that contains
constructors for all data types (Definition \ref{def.CCdatatypes}),
iterators for all data types (CBN, Definition \ref{def.CCitcbn} and
CBV, \ref{def.CCitcbv}) and arbitrarily many non-circular rules.
\end{definition}

So, $\programP$ is an ``open set'': it contains constructors and
iterators for all (infinitely many) data-types that we can define, and
it includes arbitrarily many ``non-circular rules'' that can be added
when desired. This is needed to {\em really\/} define functions using
the iteration schemes.

\subsection{Translating between call-by-name and call-by-value}\label{sec.translcbvcbn}

We can mediate between the call-by-name and the call-by-value
representations of data by defining a function 
$\StoreNat: \nattype \rightarrow \neg\neg\nattype$ and a function 
$\UnstoreNat: \neg\neg\nattype \rightarrow \nattype$. Recall from Notation
\ref{not.representation} that $\representation{n}$ is defined as $\Succ^n.\Zero$.
The function $\StoreNat$ acts
as a {\em storage operator\/} in the sense of Krivine
\cite{Krivine94} in the sense that for $t:\nattype$ with
$t \approx \representation{n}$ and $c:\nattype \rightarrow \bot$,
\[
\StoreNat.t.c \rewritestoto c.\representation{n}.
\]
So, $\StoreNat$
first evaluates the argument $t$ of type $\nattype$ completely before
passing it on to the continuation $c$.  The term $\StoreNat.t.c$ can
be defined as $\AddCBV.t.\Zero.c$, but we can also define it directly
by
\begin{eqnarray*}
\StoreNat.n.r  &\rewritesto&  n.(r.\Zero).(\StoreNatA.r)\\
\StoreNatA.r.m       &\rewritesto&  \StoreNat.m.(\StoreNatB.r)\\
\StoreNatB.r.m'      &\rewritesto&  r.(\Succ.m')
\end{eqnarray*}
It is easy to verify that $\StoreNat: \nattype\rightarrow
\neg\neg\nattype$. (Note that $\StoreNatA, \StoreNatB: \neg\nattype \rightarrow \neg
\nattype$.)

In the reverse direction, we have $\UnstoreNat:
\neg\neg\nattype \rightarrow \nattype$, defined by, given $ f : \neg\neg\nattype$,
\begin{eqnarray*}
\UnstoreNat.f.z.s &\rewritesto& f.(\UseNat.z.s)\\
\UseNat.z.s.n &\rewritesto& n.z.s
\end{eqnarray*}

Then $\UnstoreNat : \neg\neg\nattype \rightarrow \nattype$. (Note that
$\UseNat: \bot \rightarrow \neg\nattype \rightarrow \neg\nattype$.)

\begin{lemma} For all $t:\nattype$ and $n\in\IN$ with $t\oeq \representation{n}$, $\StoreNat.t.c \rewritestoto c.\representation{n}$.\\
For all $n\in \IN$, $\UnstoreNat.(\StoreNat.\representation{n}) \oeq
\representation{n}$.
\end{lemma}

\begin{proof}
For the first, we note that, if $t \oeq \representation{n}$, then: (i)
in case $n=0$, $t.z.s \rewritestoto z$; (ii) in case $n=m+1$, $t.z.s
\rewritestoto s.q$ for some $q$ with $q\oeq \representation{m}$. Then
we prove the following by induction on $n$ and $p$:
$\StoreNat.\representation{n}.(\StoreNatB^p.r) \rewritestoto
r.(\Succ^{n+p}.\Zero).$ 


For the second, we prove
$\UnstoreNat.(\StoreNat.\representation{n}).z.s \eqP
\representation{n}.z.s$, which is sufficient by Corollary
\ref{cor.eqPapprox}. 
We compute:
\begin{eqnarray*}
   \UnstoreNat.(\StoreNat.\representation{n}).z.s
&\rewritestoto& \StoreNat.\representation{n}.(\UseNat.z.s)\\
&\rewritestoto& \UseNat.z.s.\representation{n}\\
&\rewritestoto& \representation{n}.z.s
\end{eqnarray*}
Thus, $\UnstoreNat.(\StoreNat.\representation{n}).z.s =_P \representation{n}.z.s$ which was what we had to prove. \qed
\end{proof}

The map $\StoreNat$ can be seen as adding a double negation, whereas
$\UnstoreNat$ can be seen as a classical double negation law,
$\UnstoreNat : \neg\neg\nattype \rightarrow \nattype$. Note that the
fact that $\neg\neg\nattype \rightarrow \nattype$ is inhabited is not
a surprise, because $\nat$ is a negative type (ending in $\rightarrow
\bot$). The precise connection with classical logic remains to be
studied.

The storage and `unstorage' operators can most likely also be defined
for other data types.

More interesting to study further is the fact that we can combine
call-by-name and call-by-value functions. We detail this for natural
numbers.

\begin{example}\label{exa.cbvcbn}
If we have
$f_1: \nattype$ and $f_2: \nattype \rightarrow \nattype$, $c_1:\bot$,
$c_2: \nattype\rightarrow\bot$ and $n:\nattype$, then
$$\itcbn{\nattype}{\nattype} . f_1 . f_2 . n. c_1.c_2 :\bot$$ gives a
call-by-name iteration. However, one can also first define $\hatf_1 :
\neg\neg\nattype$, $\hatf_2 :\nattype \rightarrow \neg\neg\nattype$
and $\hatC : \nattype\rightarrow \bot$ by
\begin{eqnarray*}
\hatf_1. c &\rewritesto & c.f_1\\
\hatf_2. n.c &\rewritesto & c.(f_2.n)\\
\hatC.n  &\rewritesto & n.c_1.c_2
\end{eqnarray*}
Then, for $n:\nattype$, we have
$$\itcbv{\nattype}{\nattype} . \hatf_1 . \hatf_2 . \hatC.n :\bot$$
which gives call-by-value iteration. So, using this transformation
(from $f_1$ to $\hatf_1$ etc.) one can use the call-by-name functions
to compute call-by-value.
\end{example}





\section{Examples of iterators and programs}\label{examples}

\begin{example}
\label{ex:listcbn}
This is the call-by-name iterator for $\listtype{A}$ to $\nat$:
\[
	\begin{array}{lcl}
		\itcbn{\listtype{A}}{\nat} . f_1 . f_2 . x . c_1 . c_2 &
			\rewritesto &\\
\multicolumn{3}{r}{x . (\itcbnimpl{\listtype{A}}{\nat}{Nil} . f_1 . f_2 . c_1 . c_2) . (\itcbnimpl{\listtype{A}}{\nat}{Cons} . f_1 . f_2 . c_1 . c_2)} \vspace{0.2em} \\
		\itcbnimpl{\listtype{A}}{\nat}{Nil} . f_1 . f_2 . c_1 . c_2 &
			\rewritesto &
				f_1 . c_1 . c_2 \vspace{0.2em} \\
		\itcbnimpl{\listtype{A}}{\nat}{Cons} . f_1 . f_2 . c_1 . c_2 . x_1 . x_2 &
			\longrightarrow &
				f_2 . x_1 . (\itcbn{\listtype{A}}{\nat} . f_1 . f_2 . x_2) . c_1 . c_2 \\
	\end{array}
\]
\end{example}

\begin{example}
	\label{ex:listcbv}
	The scheme of Definition \ref{def.CCitcbv} yields the following call-by-value iterator for $\listtype{A}$ to $B$:
	\[
		\begin{array}{lcl}
			\itcbv{\listtype{A}}{B} . f_1 . f_2 . c . x &
				\rewritesto &
					x . (\itcbvimpl{\listtype{A}}{B}{Nil}{1} . f_1 . f_2 . c) . (\itcbvimpl{\listtype{A}}{B}{Cons}{1} . f_1 . f_2 . c) \vspace{0.2em} \\
			\itcbvimpl{\listtype{A}}{B}{Nil}{1} . f_1 . f_2 . c &
				\rewritesto &
					f_1 . c \vspace{0.2em} \\
			\itcbvimpl{\listtype{A}}{B}{Cons}{1} . f_1 . f_2 . c . x_1 . x_2 &
				\rewritesto &
					\itcbvimpl{\listtype{A}}{B}{Cons}{2} . f_1 . f_2 . c . x_2 . x_1 \vspace{0.2em} \\
			\itcbvimpl{\listtype{A}}{B}{Cons}{2} . f_1 . f_2 . c . x_2 . r_1 &
				\rewritesto &
					\itcbv{\listtype{A}}{B} . f_1 . f_2 . ( \itcbvimpl{\listtype{A}}{B}{Cons}{3} . f_1 . f_2 . c . r_1 ) . x_2 \vspace{0.2em} \\
			\itcbvimpl{\listtype{A}}{B}{Cons}{3} . f_1 . f_2 . c . r_1 . r_2 &
				\rewritesto &
					f_2 . r_1 . r_2 . c \\
		\end{array}
	\]

	We note that the number of program rules we need is highly dependent on the arity of the constructors \textbf{Nil} and \textbf{Cons}.
	Since \textbf{Nil} has no parameters, one rule is enough to define the operation on \textbf{Nil}.
	\textbf{Cons} on the other hand has two parameters.
	Because of this we get three program rules: one for evaluating each parameter of the constructor and one general rule that redirects every parameter to the corresponding program rule.
\end{example}


We now show the use of the iterators by providing the implementation of the function \textbf{Length}.
We use the iterators for $\listtype{A}$ from Example \ref{ex:listcbn} and Example \ref{ex:listcbv}.
\[
	\begin{array}{lcl}
		\mathbf{LengthCBN}.x.c_1.c_2 &
			\rewritesto &
				\itcbn{\listtype{A}}{\nattype}.\mathbf{LengthCBN}^1.\mathbf{LengthCBN}^2.x.c_1.c_2 \\
		
		\mathbf{LengthCBN}^1.c_1.c_2 &
			\rewritesto &
				\Zero.c_1.c_2 \\
		
		\mathbf{LengthCBN}^2.x.n.c_1.c_2 &
			\rewritesto &
				\Succ.n.c_1.c_2 \\
	\end{array}
\]
\[
	\begin{array}{lcl}
		\mathbf{LengthCBV}.x.c &
			\rewritesto &
				\itcbv{\listtype{A}}{\nattype}.\mathbf{LengthCBV}^1.\mathbf{LengthCBV}^2.c.x \\
		
		\mathbf{LengthCBV}^1.c &
			\rewritesto &
				c.\Zero \\
		
		\mathbf{LengthCBV}^2.x.n.c &
			\rewritesto &
				c.(\Succ.n) \\
	\end{array}
\]


We prove for $\mathbb{N}$ that the two iterator schemes (CBN and CBV)
indeed compute the desired results.  We expect that this proof can
easily be extended to prove the semantics of our iteration schemes for
any first-order algebraic data-type.  We leave this for future work.

We assume $D$ to be a data type which has a representation in
Continuation Calculus, $\mathbf{D}$, with a representation such that
$\representation{d} : \mathbf{D}$, for $d : D$. We now define what it means that a function over a data-type is represented in CC.

\weg{
We recall Definition
\ref{def:itnat} where the iterator for $\mathbb{N}$ has been defined.

Given $d : D$, $F : D \rightarrow D$, iterator $\mathbf{It}(x,F) : \mathbb{N} \rightarrow B$ is defined by, given $n \in \mathbb{N}$:
\[ 
	(\mathbf{It}(x,F))(n) = \left\{
		\begin{array}{ll}
			x & \text{if~} n = 0 \\
			F((\mathbf{It}(x,F))(n-1)) & \text{otherwise}
		\end{array}
	\right.
\]

\begin{lemma} For all $x : B$, $F : B \rightarrow B$ and $n \in \mathbb{N}$:
	\[ (\mathbf{It}(x,F))(n) = F^n(x) \]
\end{lemma}
}

\begin{definition}
We say that $f_1 :\neg\neg \mathbf{D}$ \emph{CBV-represents} $d : D$
and $f_2: \mathbf{D}\rightarrow \neg\neg\mathbf{D}$
\emph{CBV-represents} $F : D \rightarrow D$, if for all $c:\mathbf{D}
\rightarrow \bot$ and $n \in \mathbb{N}$ we have
\begin{align}
\tag{1}
f_1.c \rewritestoto ~ & c . \representation{d} \\
\tag{2}
f_2 . \representation{n} . c \rewritestoto ~ & c . \representation{F(n)}
\end{align}
\end{definition}

The following Theorem states the semantic correctness of
$\itcbv{\nattype}{D}$ in CC. 
The proof can be found in Section \ref{appendix} of the Appendix.

\begin{theorem} \label{thm.cbvit-correct}
If $f_1$ CBV-represents $d : D$ and $f_2$ CBV-represents $F : D
\rightarrow D$, then $\itcbv{\nattype}{D} . f_1 . f_2$ CBV-represents
$\mathbf{It}(d,F)$, that is: for all $c:\mathbf{D} \rightarrow \bot$
and all $n \in \mathbb{N}$ we have
\[
\itcbv{\nattype}{D} . f_1 . f_2 . c . \representation{n} ~\rewritestoto~ c . \representation{(\mathbf{It}(d,F))(n)}
\]
\end{theorem}

\begin{definition}
If $D$ has $m$ constructors, we say that $f_1 :\mathbf{D}$
\emph{CBN-represents} $d : D$ and $f_2 : \mathbf{D}\rightarrow
\mathbf{D}$ \emph{CBN-represents} $F : D \rightarrow D$, if for all
$c_i :\mathbf{D}^i$ and $n \in \mathbb{N}$ we have (writing $\vec{c} =
c_1 \ldots c_m$):
\begin{align}
\tag{1}
f_1.\vec{c} \oeq ~ & \representation{d} . \vec{c} \\
\tag{2}
f_2 . \representation{n} . \vec{c} \oeq ~ &
				\representation{F(n)} . \vec{c}
\end{align}
\end{definition}

The following Theorem states the semantic correctness of
$\itcbn{\nattype}{D}$ in CC. 
The proof can be found in Section \ref{appendix} of the Appendix.

\begin{theorem}\label{thm.cbnit-correct}
If $D$ has $m$ constructors, $f_1$ \emph{CBN-represents} $d : D$ and $f_2$ \emph{CBN-represents} $F : D \rightarrow D$, then $\itcbn{\nattype}{D} . f_1 . f_2$ CBN-represents $\mathbf{It}(d,F)$. That is, for all $c_i :\mathbf{D}^i$ and all $n \in \mathbb{N}$:
\[
\itcbn{\nattype}{D} . f_1 . f_2 . \representation{n} . \vec{c} ~\oeq~ \representation{(\mathbf{It}(d,F))(n)} . \vec{c}
\]
\end{theorem}

\section{Future Work and Conclusions}
As future work, we want to better understand the relation with
classical logic, which we have suggested in Section
\ref{examples}. Here we have also defined storage (and unstorage)
operators, which we would like to define in general for all data
types. The possibility to combine call-by-value and call-by-name in a
flexible way, which is directed by the types, is an interesting
feature, which warrants further study. The fact that computation is
completely deterministic and that the function definition of $f$
itself determines whether $f$ is cbv or cbn, makes this combining
of cbv and cbn very perspicuous.

The continuations in this paper are limited, and do not include
delimited continuations. (More examples using continuations can be
found in \cite{gg}.) It would be interesting to see if delimited
continuations can be added.

In a forthcoming paper \cite{geuvers-typlamIt} we prove the
termination of all CC terms written using the program rules of
$\programP$. This is done by translating CC with these rules to a
typed $\lambda$-calculus with (cbv and cbn) iterators. We wish to
further study the precise translations and connections between CC and
(typed) $\lambda$-calculus.

\bibliography{sources}{}
\bibliographystyle{eptcs}

\section{Appendix}\label{appendix}

Proof of Theorem \ref{thm.cbvit-correct}

\begin{proof}
	We recall the definition of $\itcbv{\nattype}{B}{B}$:
	\begin{align}
		\tag{3a}
			\itcbv{\nattype}{B} . f_1 . f_2 . c . x & \rewritesto
				x . (\itcbvimpl{\nattype}{B}{Zero}{1} . f_1 . f_2 . c) . (\itcbvimpl{\nattype}{B}{Succ}{1} . f_1 . f_2 . c) \\
		\tag{3b}
			\itcbvimpl{\nattype}{B}{Zero}{1} . f_1 . f_2 . c & \rewritesto
				f_1 . c \\
		\tag{3c}
			\itcbvimpl{\nattype}{B}{Succ}{1} . f_1 . f_2 . c . x_1 & \rewritesto
				\itcbv{\nattype}{B} . f_1 . f_2 . (\itcbvimpl{\nattype}{B}{Succ}{2} . f_1 . f_2 . c) . x_1 \\
		\tag{3d}
			\itcbvimpl{\nattype}{B}{Succ}{2} . f_1 . f_2 . c . r_1 & \rewritesto
				f_2 . r_1 . c
	\end{align}
	
	By induction on $n$ we prove that for all $n$, $P(n)$ holds, with
	\[
		P(n) ~:=~ \itcbv{\nattype}{B} . f_1 . f_2 . c . \representation{n} \rewritestoto c . \representation{(\mathbf{It}(x,F))(n)}
	\]
	
	$P(0)$ holds, because:
	\[
		\begin{array}{rcl}
			\itcbv{\nattype}{B} . f_1 . f_2 . c . \representation{0}
				& \rewritestoby{\text{(3a)}} & \representation{0} . (\itcbvimpl{\nattype}{B}{Zero}{1} . f_1 . f_2 . c) . (\itcbvimpl{\nattype}{B}{Succ}{1} . f_1 . f_2 . c) \\
				& = & \Zero . (\itcbvimpl{\nattype}{B}{Zero}{1} . f_1 . f_2 . c) . (\itcbvimpl{\nattype}{B}{Succ}{1} . f_1 . f_2 . c) \\
				& \rewritestoby{Zero} & \itcbvimpl{\nattype}{B}{Zero}{1} . f_1 . f_2 . c \\
				& \rewritestoby{\text{(3b)}} & f_1 . c \\
				& \rewritestotoby{\text{(1)}} & c . \representation{x} \\
				& = & c . \representation{(\mathbf{It}(x,F))(0)}
		\end{array}
	\]
	
	Assume $P(n)$ holds.
	
	$P(n+1)$ holds, because:
	\[
		\begin{array}{rcl}
			\itcbv{\nattype}{B} . f_1 . f_2 . c . \representation{n+1}
				& \rewritestoby{\text{(3a)}} & \representation{n+1} . (\itcbvimpl{\nattype}{B}{Zero}{1} . f_1 . f_2 . c) . (\itcbvimpl{\nattype}{B}{Succ}{1} . f_1 . f_2 . c) \\
				& = & \Succ . \representation{n} . (\itcbvimpl{\nattype}{B}{Zero}{1} . f_1 . f_2 . c) . (\itcbvimpl{\nattype}{B}{Succ}{1} . f_1 . f_2 . c) \\
				& \rewritestoby{\text{Succ}} & \itcbvimpl{\nattype}{B}{Succ}{1} . f_1 . f_2 . c . \representation{n} \\
				& \rewritestoby{\text{(3c)}} & \itcbv{\nattype}{B} . f_1 . f_2 . (\itcbvimpl{\nattype}{B}{Succ}{2} . f_1 . f_2 . c) . \representation{n} \\
				& \rewritestotoby{P(n)} & \itcbvimpl{\nattype}{B}{Succ}{2} . f_1 . f_2 . c . \representation{(\mathbf{It}(x,F))(n)} \\
				& \rewritestoby{\text{(3d)}} & f_2 . \representation{(\mathbf{It}(x,F))(n)} . c \\
				& \rewritestotoby{\text{(2)}} & c . \representation{F((\mathbf{It}(x,F))(n))} \\
				& = & c . \representation{F(F^n(x))} \\
				& = & c . \representation{F^{n+1}(x)} \\
				& = & c . \representation{(\mathbf{It}(x,F))(n+1)} \\
		\end{array}
	\]
\end{proof}

Proof of Theorem \ref{thm.cbnit-correct}

\begin{proof}
	We recall the definition of $\itcbn{\nattype}{B}$:
	\begin{align}
		\tag{3a}
			\itcbn{\nattype}{B} . f_1 . f_2 . x . \vec{c} & \rewritesto
				x . (\itcbnimpl{\nattype}{B}{Zero} . f_1 . f_2 . \vec{c}) . (\itcbnimpl{\nattype}{B}{Succ} . f_1 . f_2 . \vec{c}) \\
		\tag{3b}
			\itcbnimpl{\nattype}{B}{Zero} . f_1 . f_2 . c & \rewritesto
				f_1 . \vec{c} \\
		\tag{3c}
			\itcbnimpl{\nattype}{B}{Succ} . f_1 . f_2 . c . x_1 & \rewritesto
				f_2 . (\itcbn{\nattype}{B} . f_1 . f_2 . x_1) . \vec{c}
	\end{align}
	
	By induction on $n$ we prove that for all $n$, $P(n)$ holds, with
	\[
			P(n) ~:=~
				\itcbn{\nattype}{B} . f_1 . f_2 . \representation{n} . \vec{c} \oeq \representation{(\mathbf{It}(x,F))(n)} . \vec{c}
	\]
	
	$P(0)$ holds, because:
	\[ \begin{array}{rcl}
		\itcbn{\nattype}{B} . f_1 . f_2 . \representation{0} . \vec{c}
			& \rewritestoby{\text{(3a)}} & \representation{0} . (\itcbnimpl{\nattype}{B}{Zero} . f_1 . f_2 . \vec{c}) . (\itcbnimpl{\nattype}{B}{Succ} . f_1 . f_2 . \vec{c}) \\
			& = & \Zero . (\itcbnimpl{\nattype}{B}{Zero} . f_1 . f_2 . \vec{c}) . (\itcbnimpl{\nattype}{B}{Succ} . f_1 . f_2 . \vec{c}) \\
			& \rewritestoby{\text{Zero}} & \itcbnimpl{\nattype}{B}{Zero} . f_1 . f_2 . \vec{c} \\
			& \rewritestoby{\text{(3b)}} & f_1 . \vec{c} \\
			& \rewritestotoby{\text{(1)}} & \representation{x} . \vec{c}\\
			& = & \representation{(\mathbf{It}(x,F))(0)} . \vec{c}
	\end{array} \]
	
	Assume $P(n)$ holds.
	
	$P(n+1)$ holds, because:
	\[ \begin{array}{rcl}
		\itcbn{\nattype}{B} . f_1 . f_2 . \representation{n+1} . \vec{c}
			& \rewritestoby{\text{(3a)}} & \representation{n+1} . (\itcbnimpl{\nattype}{B}{Zero} . f_1 . f_2 . \vec{c}) . (\itcbnimpl{\nattype}{B}{Succ} . f_1 . f_2 . \vec{c}) \\
			& = & \Succ . \representation{n} . (\itcbnimpl{\nattype}{B}{Zero} . f_1 . f_2 . \vec{c}) . (\itcbnimpl{\nattype}{B}{Succ} . f_1 . f_2 . \vec{c}) \\
			& \rewritestoby{\text{Succ}} & \itcbnimpl{\nattype}{B}{Succ} . f_1 . f_2 . \vec{c} . \representation{n} \\
			& \rewritestoby{\text{(3c)}} & f_2 . (\itcbn{\nattype}{B} . f_1 . f_2 . \representation{n}) . \vec{c} \\
			& \overset{P(n), [\ref{lem.approx}]}{\oeq} & f_2 . \representation{(\mathbf{It}(x,F))(n)} . \vec{c} \\
			& \overset{\text{(2)}}{\oeq} & \representation{F((\mathbf{It}(x,F))(n))} . \vec{c} \\
			& = & \representation{F(F^n(x))} . \vec{c} \\
			& = & \representation{F^{n+1}(x)} . \vec{c} \\
			& = & \representation{(\mathbf{It}(x,F))(n+1)} . \vec{c} \\
	\end{array} \]
\end{proof}

In the proof, we say that $f_2 . (\itcbn{\nattype}{B} . f_1 . f_2 . \representation{n}) . \vec{c} ~\oeq~ f_2 . \representation{(\mathbf{It}(x,F))(n)} . \vec{c}$.
This may may not be immediately obvious, as the sub-term $\itcbn{\nattype}{B} . f_1 . f_2 . \representation{n}$ is incomplete. However, it is an immediate consequence of Lemma \ref{lem.approx}: If $M,N$ are terms of arity $k$, and $M.t_1.\CCdots.t_k =_p N.t_1.\CCdots.t_k$ for all $\vec{t}$, then $M \oeq N$.

\end{document}